\documentclass[10pt,conference]{IEEEtran}

\usepackage{amsmath,epsfig,amssymb,verbatim,amsopn,cite,subfigure,multirow}
\usepackage{balance}
\usepackage[usenames,dvipsnames]{color}
\usepackage[all]{xy} 
\usepackage{url}
\usepackage{amsfonts}
\usepackage{amssymb}
\usepackage{epsfig}
\usepackage{slashbox}
\usepackage{bm}
\usepackage{epsfig}
\usepackage{amsmath,amssymb, amstext}
\usepackage{graphicx}
\usepackage{epstopdf}

  {\proof}{\proofend}
\newtheorem{proposition}{Proposition}

\newcommand{\onecolumnfigurewidth}{0.95\linewidth}

\def\ie{\emph{i.e}.}

\newcounter{mytempeqcounter}

\newcommand{\bigformulatop}[2]{%
  \begin{figure*}[!t]
    \normalsize
    \setcounter{mytempeqcounter}{\value{equation}}
    \setcounter{equation}{#1}
    #2

    \setcounter{equation}{\value{mytempeqcounter}}
    \hrulefill
    \vspace*{4pt}
  \end{figure*}
}

\newcommand{\Hsr}{h_{s,r_l}}
\newcommand{\Hrd}{h_{r_l,d}}

\newcommand{\YR}{y_{s,r_l}}
\newcommand{\YD}{y_{r_l,d}}

\newcommand{\AuthorOne}{Ashkan Kalantari}
\newcommand{\AuthorTwo}{Mohammadali Mohammadi}
\newcommand{\AuthorFour}{Mehrdad Ardebilipour}

\definecolor{light-gray}{gray}{0.65}

\newcommand{\ThankOne}{Author1, Author2, Author3 Author4 and Author5 are with
. email:\{email1, email2, email3, email4\}@email.com.}

\title{Performance Analysis of Opportunistic Relaying Over Imperfect Non-identical Log-normal Fading Channels}
\author{\authorblockN{\AuthorOne,\:\AuthorTwo,\:and\:\AuthorFour\thanks{\ThankOne}}
K.~N.~Toosi University of Technology, Tehran, Iran,\\
Email:\texttt{ \{A.Kalantari; M.A.Mohammadi\}@ee.kntu.ac.ir},
\texttt{ Mehrdad@eetd.kntu.ac.ir} }

\begin{document}
\maketitle

\begin{abstract}

Motivated by the fact that full diversity order is achieved using
the "best-relay" selection technique, we consider opportunistic
amplify-and-forward and decode-and-forward relaying systems. We
focus on the outage probability of such a systems and then derive
closed-form expressions for the outage probability of these systems
over independent but non-identical imperfect Log-normal fading
channels. We consider the error of channel estimation as a Gaussian
random variable. As a result the estimated channels distribution are
not Log-normal either as would be in the case of the Rayleigh fading
channels. This is exactly the reason why our simulation results do
not exactly matched with analytical results. However, this
difference is negligible for a wide variety of situations.

\end{abstract}

\section{Introduction}

Cooperative communication increases the performance quality of
communication systems in terms of capacity, outage probability and
symbol error probability (SEP) dramatically. Two main relaying
protocols that have been researched a lot, are amplify-and-forward
(AF) and decode-and-forward
(DF)~\cite{Laneman:IT:Sep:2003,Laneman:IT:Dec:2004}.
In~\cite{Hasna:COM:2003,Hasna:WCOM:2003,Hasna:COM:2004,Kaveh:WCOM:2004}
performance of AF and DF protocols in terms of outage probability
and SEP have been widely investigated over Nakagami-\emph{m} and
Rayleigh fading channels.

In cooperative communication networks, the use of multiple relays to
facilitate the source-destination communication was proposed to
increase the spatial diversity
gain~\cite{Ribeiro:WCOM:2005,Ikkii:Wcom:2009}. To avoid the
interfering, source and all the relay transmissions must take place
on orthogonal channels. Thus multiple relay cooperation is
considered inefficient in terms of channel resources and bandwidth
utilization. To overcome this problem, opportunistic relaying (OR)
has been proposed at which only the best-relay from a set of $L$
available relays is selected to participate in
communication~\cite{Bletesas:JSAC:2006,Bletsas:WCOM:2010}. It was
shown that OR achieves full diversity order. With this technique,
the selection strategy is to choose the relay with the best
equivalent end-to-end channel gain which is obtained as the highest
minimum of the channel gains of the first and the second hops under
DF protocol or with the best harmonic mean of both channel gains
under AF protocol~\cite{TRUNG:COML:2009,Yi:COML:2006}.

In~\cite{Kostic:2005,Uysal:WCOm:2008,Renzo:Com:2010}
performance of Log-normal fading channels over different structures
and relaying protocols has been investigated. From the practical
point of view, Log-normal distribution is encountered in many
communication scenarios. For instance, when indoor communication is
used at which users are moving, Log-normal distribution not only
models the moving objects, but also the reflection of the bodies.
Moreover, it models the action of communicating with robots in a
closed environment like a factory~\cite{Renzo:Glob:2008}. In indoor
radio propagation environments, terminals with low mobility have to
rely on macroscopic diversity to overcome the shadowing from indoor
obstacles and moving human bodies. Indeed, in such slowly varying
channels, the small-scale and large-scale effects tend to get mixed.
In this case, Log-normal statistics accurately describe the
distribution of the channel path gain~\cite{Alouini:Com:2002}.

A main underlying assumption in majority of the current literature
on cooperative communication is the availability of the channel
state information (CSI) at the receiver. Recently there has been an
interest in evaluation the performance of relay networks over
imperfect channels
\cite{patel:JVT:2006jan,SEu:TVT:2009,Seify:Glob:2010}. Performance
analysis of opportunistic AF and DF schemes over Log-normal channels
is a non-trivial task when the CSI is imperfectly known at all
nodes. To the best of our knowledge there have been no reported
results on the outage probability analysis of such a systems yet.
The main contribution of this paper is to derive closed-form
expressions for the outage probability of the multi-relay DF and AF
systems, employing ``best-worst'' and ``best harmonic mean'' relay
selection criteria over imperfect Log-normal fading channels,
respectively.

The rest of the paper is organized as follows:
Section~\ref{sec:System_Model} describes the system model.
Section~\ref{sec:Model Estim}, discusses the Gaussian error model
for imperfect channel estimation. We take advantage of this model in
our relay selection scenario to obtain output instantaneous SNR.
Section~\ref{sec:Perfo Analy} provides the harmonic mean of two
Log-normal random variables (R.V) for AF and the equivalent CDF of
the best-worse selection criterion for DF in order to calculate
outage probability. Section~\ref{sec:Num Res} presents simulation
results, while Section~\ref{sec:con} provides some concluding
remarks.

\section{System Model}
\label{sec:System_Model}

We consider a multi-relay scenario, in which a source node (S)
communicates to a  destination node (D) via  multiple fixed relays $
(R_l, l=1,...,L)$. We assume that there is no direct link between
the source and the destination, and communication occurs using a
two-hop protocol over two time slots\cite{Laneman:IT:Dec:2004}. The
fading coefficient over source to relay $l$, and relay $l$ to
destination are denoted with $\Hsr$ and $\Hrd$ which are assumed to
be independent and non-identically distributed Log-normal R.Vs.
Dropping the indexes, channel gains during the transmission of a bit
are modeled by $h=10^{0.1X}$, at which $X\sim
\mathcal{N}(\mu,\sigma^2)$ and $\mathcal{N}$ denotes a Gaussian
distribution. During the first time slot, source broadcasts the
signal to $L$ relays. In the second time slot, only the best-relay
forwards the signal to the destination, and source remains idle. Let
us denote with $P_s$ the power transmitted by the source and thus
the set of below equations summarize the operation taking place for
each symbol
\begin{align}
  \YR &= \sqrt{P_s} x \Hsr+n_{r_l} \quad l=1,\cdots,L \label{eqn:S_R_signal}\\
  \YD &=  x_r \Hrd+n_d,\label{eqn:R_D_signal}
\end{align}
where $x$ is the transmitted signal with power $E[\vert x \vert
^2]=1$, 
and $n_{r,l}$ and $n_d$ are complex additive white Gaussian noise
(AWGN) in the relay and destination, respectively. Without loss of
generality, we assume that all the AWGN terms have equal variance as
$N_0$. For AF relaying $x_r=A_l\YR$ at which relay amplification
factor, $A_l$, is chosen to satisfy an average power constraint and
will be defined later. For DF relaying $x_r=\sqrt{P_r}\hat{x}$,
where $\hat{x}$ is obtained after demodulating $\YR$ followed by
modulating for retransmit to the destination. Therefore, in brief we
study the OR~\cite{Bletesas:JSAC:2006} with two conventional
relaying strategies at the relay:
\begin{itemize}
\item \textbf{DF}: The best-relay decodes the
message, re-encodes it and transmits that message in next time slot.
\item \textbf{AF} : In the second time slot the best-relay process  the received signal and forwards it to the destination.
\end{itemize}

In the sequel we investigate the performance of OR-DF and OR-AF
schemes over Log-normal fading channels where relay and destination
are provided with a estimation of their corresponding channels.

\subsection{Transmission with DF Protocol}

In the second time slot, for DF scheme, only the best-relay
according to the best-worse criterion~\cite{Bletesas:JSAC:2006} is
chosen to decode and re-encode the received signal which yields
$\hat{x}$. Then the selected $l^{th}$ relay send
$x_r=\sqrt{P_r}\hat{x}$ to the destination where $P_r$ is the relay
power. As a result the received signal at the destination is given
by
\begin{align}
  \YD &=  \sqrt{ P_r }\hat{ x } \Hrd+n_d,\label{eqn:R_D_DF}
\end{align}

\subsection{Transmission with AF Protocol}

Under AF protocol the relay with the best harmonic mean of both
source to relay and relay to destination gains is chosen to forward
$x_r=A_l\YR$ to destination~\cite{Hasna:WCOM:2003}. Note that $A_l=
\sqrt{\frac{ P_r }{P_s \vert \hat{h}_{s,r_l} \vert^2 + N_0}}$,
asserts the relay amplification factor which controls the output
power of the relay~\cite{Laneman:IT:Sep:2003}. Since $A_l$ depends
on the fading coefficient, each relay has to estimate its own
received channel. We assume that relays estimate their corresponding
channels, $\hat{h}_{s,r_l}$, and then use it to amplify the received
signal. The received signal at the destination is of the form
\begin{align}
  \YD &=  A_l\YR \Hrd+n_d \nonumber\\
  &=A_l\Hsr\Hrd x+ A_l\Hrd n_{r_l}+n_d.\label{eqn:R_D_AF}
\end{align}
Armed with these system models, in the consecutive sections we model
the imperfect CSI at the receiving nodes, and then study the performance
of aforementioned AF and DF schemes with relay selection, in term of
outage probability.

\section{Instantaneous SNR with Imperfect CSI}
\label{sec:Model Estim}

We denote the estimated and exact channel coefficients as $\hat{h}$
and $h$, respectively. To estimate the h linearly with respect to $\hat{h}$, we
employ the following model~\cite{Gu:2003}
\begin{align}
  h =\rho\hat{h}+e,\label{HHa}
\end{align}
where $e$ is the channel estimation error modeled by zero mean
complex Gaussian distribution with variance $\sigma_e^2$ and $\rho$
is the correlation coefficient between $h$ and $\hat{h}$ which is
given by $\rho=\frac{\sigma_e^2}{\sigma_{\hat{h}}^2}$~\cite{Gu:2003}. The
variances of the error and the exact channel coefficient are related
by $\rho$ as
\begin{align}
 \sigma^2_e=(1-\rho)\sigma^2_h.\label{eqn:Poutage}
\end{align}
 Since estimated channel is the combination of a Log-normal and a complex Gaussian R.V, it's
 pdf does not exactly follow a Log-normal distribution. However, in Section~\ref{sec:Num Res}
 we will show that when the correlation coefficient between the estimated channel and the real one  is near to
 one (\ie, $\rho\approx 1$), this approximation is acceptable. By employing the imperfect channel estimations at the receiving
 node we obtain a closed-form expressions for instantaneous SNR at the destination.

\subsection{Transmission with DF Protocol}

Substituting \eqref{HHa} into \eqref{eqn:S_R_signal} and \eqref{eqn:R_D_signal}, we have
\begin{align}
  \YR\! \!&=\!\! \sqrt{P_s}x\rho_{s,r_l}\hat{h}_{s,r_l}\!\!+\!\sqrt{P_s} x e_{s,r_l}\!\! + \!n_{r_l}\quad l=1,\cdots,L \label{eqn:S_R_signal estimate}\\
  \YD &=  x_r \rho_{r_l,d}\hat{h}_{s,r_l}+ x_r e_{r_l,d} + n_d.\label{eqn:R_D_signal estimate}
\end{align}
After receiving the signal in the $l^{th}$ relay, since the noise
power is not the same on all sub-channels, each diversity branch has
to be weighted by its corresponding complex fading gain over total
noise power on that particular branch. Therefore, the selected relay
and the destination will decode the received signal using
MRC as~\cite{Kaveh:WCOM:2004} 
\begin{align}
  \hat{y}_{s,r_l} = \frac{{\hat{h}}^{\ast}_{s,r_l}}{N_0}\YR, \quad  \hat{y}_{r_l,d} = \frac{{\hat{h}}^{\ast}_{r_l,d}}{N_0}\YD, \label{eqn:SNR_MRC}
\end{align}
then using~\eqref{eqn:SNR_MRC}, instantaneous SNR at the relay and destination can be formulated as
\begin{align}
    &\hat{\gamma}_{s,r_l}^{eff}=\tau_{s,r_l} \frac{P_s}{N_0} \vert \hat{h}_{s,r_l} \vert^2=\tau_{s,r_l} \bar{\gamma}_{s,r_l} \vert \hat{h}_{s,r_l} \vert^2=
    \tau_{s,r_l}\hat{\gamma}_{s,r_l},\\
    &\hat{\gamma}_{r_l,d}^{eff}\!=\!\tau_{r_l,d} \frac{P_r}{N_0} \vert \hat{h}_{r_l,d} \vert^2=\tau_{r_l,d} \bar{\gamma}_{r_l,d} \vert \hat{h}_{r_l,d} \vert^2=  \tau_{r_l,d}\hat{\gamma}_{r_l,d},\label{eqn:S_R_signal SNR eff}
\end{align}
where $\tau_{s,r_l}=\frac{\rho_{s,r_l}^2 \vert
\hat{h}_{s,r_l} \vert^2}{1+\frac{P_s \sigma_{e_{s,r_l}}^2}{N_0}}$ and $\tau_{r_l,d}=\frac{\rho_{r_l,d}^2 \vert \hat{h}_{r_l,d}
\vert^2}{1+\frac{P_s \sigma_{e_{r_l,d}}^2}{N_0}}$.
\subsection{Transmission with AF Protocol}

Remembering that our transmission model for AF scheme is given
by~\eqref{eqn:R_D_AF}, substituting \eqref{HHa} into
\eqref{eqn:R_D_AF}, the received signal in the destination is
\begin{align}
\YD &=\left[\sqrt{ P_s } A_l x (
\rho_{s,r_l}\rho_{r_l,d}\hat{h}_{s,r_l}\hat{h}_{r_l,d})\right]
  +\nonumber\\
  &\quad~ \left[{\sqrt{ P_s }A_l x (\rho_{s,r_l} \hat{h}_{s,r_l} e_{s,r_l}+
  \rho_{r_l,d} \hat{h}_{r_l,d} e_{r_l,d} +
  e_{s,r_l} e_{r_l,d} )+}\right.\nonumber\\
  &\left.\quad ~{ A_l\rho_{r_l,d} \hat{h}_{r_l,d} e_{r_l,d}}\right] + \left[{A_l\rho_{r_l,d} \hat{h}_{r_l,d} n_r + n_d }\right], \label{eqn:R_D_AF S E N}
\end{align}
where, the first term presents the received signal, the second and
third terms stand for the error signal, and the overall noise at the
destination is $\tilde{n}_d\triangleq {A_l\rho_{r_l,d}
\hat{h}_{r_l,d}n_r + n_d}$ which is a complex Gaussian R.V with
$\tilde{n}_d\sim \mathcal{N}(0,\sigma_{\tilde{n}}^2)$ where
$\sigma_{\tilde{n}}^2$ is
\begin{align}
\sigma_{\tilde{n}}^2 = N_0 \left( 1 + \frac{P_r \rho_{r_l,d} \vert
\hat{h}_{r_l,d} \vert^2 }{P_s \vert \hat{h}_{s,r_l} \vert^2 + N_0}
\right) .\label{eqn:n_tot var }
\end{align}
Using MRC at the input of destination, the estimated signal
is\vspace{-0.2em}
\begin{align}
\hat{y}_{r_l,d}=\frac{\hat{h}^\ast_{s,r_l}
\hat{h}^\ast_{r_l,d}}{\sigma_{\tilde{n}}^2}\YD. \label{eqn:MRC AF }
\end{align}
Supposing that $n_r$, $n_d$, $e_{s,r_l}$ and $e_{r_l,d}$ are
processes that are independent from each other, the instantaneous
SNR at the destination is obtained as~\eqref{eqn:SNRdestAF}, at the
top of the next page. In order to have a more tractable form, we
neglect $N_0^2( \frac{P_s}{N_0}\frac{P_r}{N_0}\sigma_{e_{s,r_l}}^2
\sigma_{e_{r_l,d}}^2 + 1)$ in \eqref{eqn:SNRdestAF}. So, we can
further simplify \eqref{eqn:SNRdestAF} as \vspace{-0.3em}
\bigformulatop{14}{ \vskip-0.5cm
\begin{align}
\hat{\gamma}_{d,l}^{eff}=\frac{ P_s P_r \rho_{s,r_l} \rho_{r_l,d}  \vert
\hat{h}_{r_l,d} \vert^2 \vert \hat{h}_{s,r_l} \vert^2 }{N_0^2\left(
\frac{P_s}{N_0} \vert \hat{h}_{s,r_l} \vert^2 (1+\rho_{s,r_l}^2
\frac{P_r}{N_0}) \sigma_{e_{r_l,d}}^2 + \frac{P_r}{N_0} \vert
\hat{h}_{r_l,d} \vert^2 (1+\rho_{r_l,d}^2 \frac{P_s}{N_0})
\sigma_{e_{s,r_l}}^2 + \frac{P_s}{N_0}  \sigma_{e_{s,r_l}}^2
\frac{P_r}{N_0} \sigma_{e_{r_l,d}}^2 + 1  \right)}.\label{eqn:SNRdestAF}
\end{align}
}
\setcounter{equation}{15}
\begin{align}
\hat{\gamma}_{d,l}^{eff}=\frac{ \rho_{s,r_l}^2 \rho_{r_l,d}^2 \hat{\gamma}_{s,r_l} \hat{\gamma}_{r_l,d}  }
{ \hat{\lambda}_{s,r_l} \hat{\gamma}_{s,r_l} + \hat{\lambda}_{r_l,d} \hat{\gamma}_{r_l,d} }
,\label{eqn:SNR D AF Short}
\end{align}
where $\lambda_{s,r_l}\!\!=\!1 \!+\! \rho_{s,r_l}^2
\epsilon_{r_l,d}$, $\lambda_{r_l,d}\!=\!1 \!+ \rho_{r_l,d}^2
\epsilon_{s,r_l}$, $\hat{\gamma}_{s,r_l}\!=\!\frac{P_s \vert
\hat{h}_{s,r_l} \vert^2}{N_0}$,
$\hat{\gamma}_{r_l,d}\!\!=\!\!\frac{P_r \vert \hat{h}_{r_l,d}
\vert^2}{N_0}$, $\epsilon_{s,r_l}\!\! =\!\! \frac{P_s
\sigma_{e_{s,r_l}}^2 }{N_0}$, and $\epsilon_{r_l,d} \!\!=\!\!
\frac{P_r \sigma_{e_{r_l,d}}^2 }{N_0}$.

\section{Performance Analysis}
\label{sec:Perfo Analy}

In this section we derive closed-form expressions for outage
probability of the OR scheme under AF and DF protocols. Outage
probability is defined as the probability that the instantaneous SNR
at the receiver, $\gamma$, falls below a predetermined protection
ratio, $\gamma_{th}$, namely
\begin{align}
      P_{out}=P[\gamma\leq\gamma_{th}]=\int_{0}^{\gamma_{th}}
      f_{\Upsilon}(\gamma)d\gamma.\label{P_out}
\end{align}
where $f_{\Upsilon}(\gamma)$ represent the pdf of the instantaneous
SNR. It can readily be seen that the outage probability is actually
the cumulative distribution function (CDF) of $\gamma$ evaluated at
$\gamma_{th}$. Before proceeding, we introduce following theorems
from~\cite{Papoulis:1984} which will be used in the sequel to derive
$f_{\Upsilon}(\gamma)$.

Theorem 1: If X and Y are two R.Vs with relation $Y=mX$, then
$f_Y(\gamma)=\frac{1}{m}f_X(\frac{\gamma}{m})\label{Y=MX}$.

Theorem 2: If $X$ is a Log-normal R.V with distribution
$\textit{X}\sim\textit{LogN}(\mu_X,\sigma_X^2)$, then
$Y=\frac{1}{X}$ is a Log-normal R.V distributed as
$\textit{Y}\sim\textit{LogN}(-\mu_X,\sigma_X^2)$.

Theorem 3: If $X$ is a Log-normal R.V with distribution
$\textit{X}\sim\textit{LogN}(\mu_X,\sigma_X^2)$, then $Y=X^2$ is a
Log-normal R.V with distribution defined as
$\textit{Y}\sim\textit{LogN}(2\mu_X,4\sigma_X^2)$.

Theorem 4: If $X$ is a Log-normal R.V with distribution $\textit{X}\sim\textit{LogN}(\mu_X,\sigma_X^2)$, then
considering theorem 1,  $Y=mX$ is a Log-normal R.V with distribution defined as
 $\textit{Y}\sim\textit{LogN}(\mu_X+10\log(m),\sigma_X^2)$.

\subsection{Transmission with DF Protocol}
If we suppose that $\hat{h}_{s,r_l}$ and $\hat{h}_{r_l,d}$, the
estimations of the source-relay and relay-destination channels,
respectively are available at destination, in the first phase,
destination node equipped with selection combiner (SC), selects the
worst hop of each branch as\vspace{-0.5em}
\begin{align}
    \hat{\gamma}_{eq_l}=\min\{\hat{\gamma}_{s,r_l},\hat{\gamma}_{r_l,d}\}.\label{HOP SC}
\end{align}
In second phase, SC selects the branch with the best
$\hat{\gamma}_{eq_l}$ as
\begin{align}
   \hat{\gamma}_{SC}=\max\{\hat{\gamma}_{eq_1},\hat{\gamma}_{eq_2},\ldots,\hat{\gamma}_{eq_L}\}.\label{Bra SC}
\end{align}
Since the SC chooses the weakest part of each branch and then the best
one is selected to send the signal, the occurrence of outage is
equal to the case when the best weak link's SNR is under the
threshold ($\gamma_{th}$),
\begin{align}
    P_{out}=P( \hat{\gamma}_{SC}\leq\gamma_{th}).\label{OUT SC}
    \end{align}

Without loss of generality, we stipulate equal power allocation to
source and best-relay ($P_s = P_r =P$). By considering the pdf of Log-normal
R.V~\cite{M.K.SimonandM.S.Alouini;2005}, $\hat{h}$, with corresponding Normal parameters
defined as $\mu_{\hat{h}}$ and $\sigma_{\hat{h}}^2$, then by using theorems 3 and
4, respectively, after some elementary manipulations and dropping the indexes, the pdf of
$\hat{\gamma}=\frac{P}{{{N_0}}}{\left| \hat{h} \right|^2}=\bar{\gamma}{\left| \hat{h} \right|^2}$ is obtained as
\begin{align}
    &f_{\hat{\Upsilon}}(\hat{\gamma})=\frac{\xi}{\sqrt{2\pi}\sigma_{\hat{\Upsilon}} \hat{\gamma}} \exp\left[-\frac{(10\log_{10}\hat{\gamma}-\mu_{\hat{\Upsilon}})^2}{2\sigma_{{\hat{\Upsilon}}}^2}\right] \label{pdf log}    \\
  &\mu_{\hat{\Upsilon}}\!=\!10\log_{10}\!E(\hat{\gamma})\!\!-\!\!5\log_{10}\!\Psi(\hat{\gamma}), \quad\sigma_{\hat{\Upsilon}}\!=\!\frac{100}{\ln10}\Psi(\hat{\gamma}),
   \label{LogN_paramet}
    \end{align}

where, $\xi=\frac{10}{\ln10}$,
$\textit{$\mu_{\hat{\Upsilon}}$}=2\mu_{\hat{h}}+10\log\bar{\gamma}$,
$\textit{$\sigma_{\hat{\Upsilon}}^2$}=4\sigma_{\hat{h}}^2$,
$\Psi(\hat{\gamma})=\left( 1 + Var(\hat{\gamma})/E(\hat{\gamma})^2
\right)$. Relations \eqref{LogN_paramet} help us derive the
parameters of Log-normal distribution directly from the variable,
$\hat{\gamma}$. We also can express the CDF of $\hat{\gamma}$
as~\cite{Alouini:Com:2002}
\begin{align}
    F_{\hat{\Upsilon}}(\hat{\gamma})=Q\left(\frac{\mu_{\hat{\Upsilon}}-10\log_{10}\hat{\gamma}}{\sigma_{\hat{\Upsilon}}}\right),
    \quad \hat{\gamma}\geq0
     \label{LogN_SNR CDF}
    \end{align}
where,
$Q(x)=\frac{1}{\sqrt{2\pi}}\int_{x}^{\infty}e^{-\frac{u^2}{2}}du$ is
the standard one-dimensional Gaussian function. According to
independency of channel gains, the pdf of $\hat{\gamma}_{eq_l}$
in~\eqref{HOP SC} can be expressed as
\begin{align}
   F_{\hat{\Upsilon}_{eq_l}}(\hat{\gamma})&=
   P\left (\min\{\hat{\gamma}_{s,r_l},\hat{\gamma}_{r_l,d}\leq\hat{\gamma}\}\right )\nonumber\\
   &=1-(1-F_{\hat{\Upsilon}_{s,r_l}}(\hat{\gamma}))(1-F_{\hat{\Upsilon}_{r_l,d}}(\hat{\gamma})).\label{gamma eq}
    \end{align}

Substituting \eqref{LogN_SNR CDF} into \eqref{gamma eq}, yields
\begin{align}
   F_{\hat{\Upsilon}_{eq_l}}(\hat{\gamma})&=1-(1-Q\left(\Omega_{s,r_l}\right))(1-Q\left(\Omega_{r_l,d}\right)). \label{F gamma}
    \end{align}
where
$\Omega_{a,b}=(\mu_{\hat{\Upsilon}_{a,b}}-10\log_{10}\hat{\gamma})/\sigma_{\hat{\Upsilon}_{a,b}}$
for $a \in \{s,r_l\}$ and $b \in \{r_l,d\}$ . According to the fact
that our branches are independent and using \eqref{Bra SC},
$P_{out}$ is given by
\begin{align}
   P_{out}\!=\!P(\hat{\gamma}_{eq_1}\!\leq\!\gamma_{th})P(\hat{\gamma}_{eq_2}\!\leq\!\gamma_{th})\cdots\ P(\hat{\gamma}_{eq_L}\leq\gamma_{th}). \label{P out s}
    \end{align}
Substituting \eqref{F gamma} into \eqref{P out s} yields outage
probability for OR-DF scheme as
\begin{align}
   P_{out}=\prod_{l=1}^L \left( Q(\Omega_{s,r_l})\! +\! Q(\Omega_{r_l,d})\!-\! Q(\Omega_{s,r_l})Q(\Omega_{r_l,d})  \right)
   . \label{P out}
    \end{align}

\subsection{Transmission with AF Protocol}

In the OR-AF scenario, destination selects the maximum
harmonic mean of both S-R and R-D channel gains. As a result the
outage probability can be readily obtained as
\begin{align}
   P_{out}\!=\!P\left( \hat{\gamma}_{SC} \!\!=\! \max\{ \hat{\gamma}_{d,1}^{eff},\hat{\gamma}_{d,2}^{eff},\cdots,
   \hat{\gamma}_{d,L}^{eff} \}\leq \gamma_{th} \right), \label{P out AF}
    \end{align}
Similar to DF mode, using the independency of the branches, we arrive
at
\begin{align}
   P_{out}\!=\!P\!\left(\hat{\gamma}_{d,1}^{eff}\!\!\leq \!\!\gamma_{th}\right)\!\!P\!\!\left(\hat{\gamma}_{d,2}^{eff}\!\!\leq \!\!\gamma_{th}\right)\cdots
   \!P\!\!\left(\hat{\gamma}_{d,L}^{eff}\!\!\leq\!\! \gamma_{th}\right), \label{P out AF exp}
    \end{align}
which is the product of CDFs of $\hat{\gamma}_{d,l}^{eff},
\l=1,\cdots,L$. The CDF of $\hat{\gamma}_{d,l}^{eff}$ is given in
following preposition.

\begin{proposition} The pdf of the received instantaneous SNR at the destination
for OR-AF relaying protocol over imperfect non-identical Log-normal
fading channels is $\hat{\gamma}_{d,r_l}^{eff}\sim\textit{ LogN }(
-\mu_{\chi} + 10\log\frac{ \rho_{s,r_l}^2 \rho_{r_l,d}^2 } {
\lambda_{s,r_l} \lambda_{r_l,d} },\sigma_{\chi}^2 )$.
\end{proposition}

\begin{proof}We can rewrite~\eqref{eqn:SNR D AF Short} as
\begin{align}
   \hat{\gamma}_{d,r_l}^{eff} = \frac{ \rho_{s,r_l}^2 \rho_{r_l,d}^2 } { \lambda_{s,r_l} \lambda_{r_l,d} }
   \left( \frac{ 1 }{ \lambda_{s,r_l} \hat{\gamma}_{s,r_l} } +
   \frac{ 1 }{ \lambda_{r_l,d}
\hat{\gamma}_{r_l,d} } \right)^{-1}.
   \label{eqn:S_R_signal SNR eff af}
    \end{align}
In order to simplify the analysis we introduce a new random variable
$\chi$ given by the summation of two random variables as
$\chi=\alpha+\beta$, where $\alpha=\frac{ 1 }{ \lambda_{s,r_l}
\hat{\gamma}_{s,r_l} }$ and $\beta=\frac{ 1 }{ \lambda_{r_l,d}
\hat{\gamma}_{r_l,d} }$ . Remembering that the pdf of the S-R and
R-D channels are given by $h_{s,r_l}\sim\textit{ LogN }(
\mu_{s,r_l},\sigma_{s,r_l}^2 )$ and $h_{r_l,d}\sim\textit{ LogN }(
\mu_{r_l,d},\sigma_{r_l,d}^2 )$, it become obvious that both
$\alpha$ and $\beta$ have a Log-normal distribution as
$\alpha\sim\textit{ LogN }( -(2\mu_{s,r_l} +
10\log(\bar{\gamma}_{s,r_l}\lambda_{s,r_l}) ),4\sigma_{s,r_l}^2 )$ and $\beta\sim\textit{
LogN }( -(2\mu_{r_l,d} + 10\log(\bar{\gamma}_{r_l,d}\lambda_{r_l,d})),4\sigma_{r_l,d}^2
)$. Next we derive the pdf of $\chi$. For this purpose we employ the
Wilkinson method which is described in~\cite{Schwar:Bell:1982} as
follows:

\emph{Wilkinson method}: If $X_1,\cdots, X_N$ are $N$ Log-normal
R.Vs, then $Z=X_1+X_2+\cdots+X_N$, can be approximated by a
Log-normal R.V with the following parameters
\small
\begin{align}
   \mu_Z\!=\!\frac{1}{\lambda}\!\left(\!2\ln(u_1)\!-\!\frac{1}{2}\ln(u_2)\!\right)\!,\quad  \sigma_Z\!=\!\frac{1}{\lambda}\sqrt{ \ln(u_2)\!-\!2\ln(u_1)  }\label{eqn:wilki variance}
    \end{align}
\normalsize
\vspace{-0.3em}
\begin{align}
  & u_1= \sum_{i=1}^N e^{(\mu_{X_i} + \frac{\sigma_{X_i}^2}{2} )}= e^{(\mu_{Z} + \frac{\sigma_{Z}^2}{2} )}\label{eqn:wilki u1} \\
  & u_2= \sum_{i=1}^N e^{(2\mu_{X_i} + 2\sigma_{X_i}^2 )} + 2 \sum_{i=1}^{N-1}  \sum_{j=i+1}^N  e^{(\mu_{X_i} + \sigma_{X_i}^2  )}\times\nonumber\\
  &\quad  e^{ \frac{1}{2} ( \sigma_{X_i}^2 + \sigma_{X_j}^2 + 2r_{ij} \sigma_{X_i}^2 \sigma_{X_j}^2  )}=
    e^{(2\mu_{Z} + 2\sigma_{Z}^2 )},\label{eqn:wilki u2}
\end{align}
where $\lambda=\frac{\ln(10)}{10}$ and $r_{i,j}$ is the correlation coefficient between $X_i$ and $X_j$
which is defined as
\begin{align}
r_{i,j}=\frac{  E \{( X_i - \mu_{X_i} )( X_j - \mu_{X_j} )\}  }{ \sigma_{X_i} \sigma_{X_j} }.\label{eqn:coorelation coe}
    \end{align}
Replacing $\mu_{X_1}$, $\sigma_{X_1}$, $\mu_{X_2}$ and
$\sigma_{X_2}$ with $\mu_{\alpha}$, $\sigma_{\alpha}$, $\mu_{\beta}$
and $\sigma_{\beta}$ and using \eqref{eqn:wilki variance},
$\mu_{\chi}$ and $\sigma_{\chi}^2$ are obtained. Finally, using
theorem 2 and 1, respectively. It can be readily checked that
$\hat{\gamma}_{s,r_l}^{eff}$ has a Log-normal distribution as
$\hat{\gamma}_{s,r_l}^{eff}\sim\textit{ LogN }( -\mu_{\chi} +
10\log\frac{ \rho_{s,r_l}^2 \rho_{r_l,d}^2 } { \lambda_{s,r_l}
\lambda_{r_l,d} },\sigma_{\chi}^2 )$. As $N$ becomes large, the
central limit theorem states that the sum will become close to
Gaussian distribution.
\end{proof}

\begin{figure}[h]
\centering
\includegraphics[width=\onecolumnfigurewidth]{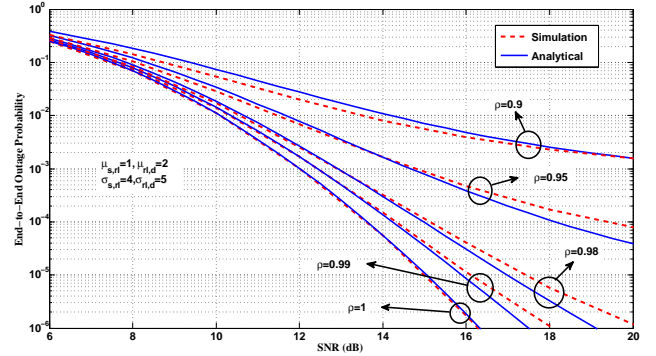}
\vspace{-1mm} \caption{ Outage probability for OR-DF mode over
imperfect Log-normal fading channels for different $\rho$.  }
\label{fig: OR-DF} \vspace{-2mm}
\end{figure}
\begin{figure}[h]
\centering
\includegraphics[width=\onecolumnfigurewidth]{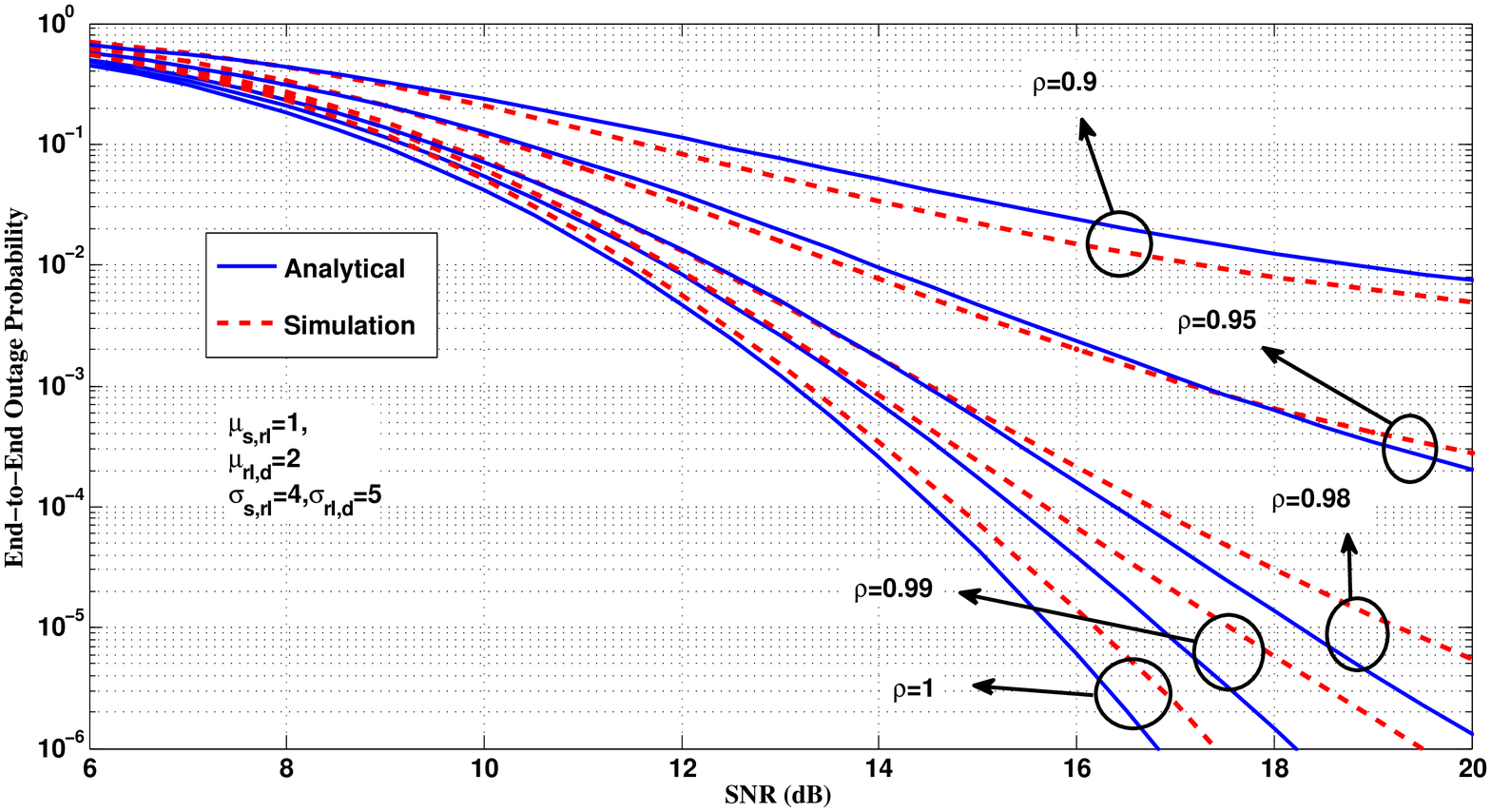}
\vspace{-1mm} \caption{ Outage probability for OR-AF mode over
imperfect Log-normal fading channels for different $\rho$.  }
\label{fig: OR-AF} \vspace{-4mm}
\end{figure}

Using \eqref{LogN_SNR CDF}, the outage probability of the
OR-AF relaying over imperfect non-identical fading
channels can be expressed as
\begin{align}
  P_{out}= \prod_{l=1}^L Q(\frac{\mu_{\hat{\gamma}_{d,r_l}^{eff}}-10\log_{10}\hat{\gamma}}{\sigma_{\hat{\gamma}_{d,r_l}^{eff}}}). \label{P out AF 3}
    \end{align}

\section{Numerical Results}
\label{sec:Num Res}
In order to justify our analytical results, we provide some
monte-carlo simulations in this section. In all simulations, unless
it mentioned otherwise, we have three relays ($L=3$),
$\gamma_{th}=3$,
and $\rho_{s,r_l}=\rho_{r_l,d}$. Log-normal fading channel
parameters and the correlation coefficient between the exact and
estimated channel, $\rho$, are given in each figure. We can
calculate the variance of  the each channel ($\sigma_h^2$)
using~\cite{Papoulis:1984}
\begin{align}
&\sigma _h^2 = \mu _X^2({10^{\sigma _X^2{\textstyle{{\ln 10} \over {100}}}}} - 1),\label{Lognormal var}
    \end{align}
    \normalsize
where, $\mu_X$ and $\sigma_X^2$ are the Log-normal fading parameters.

\begin{figure}[h]
\centering
\includegraphics[width=\onecolumnfigurewidth]{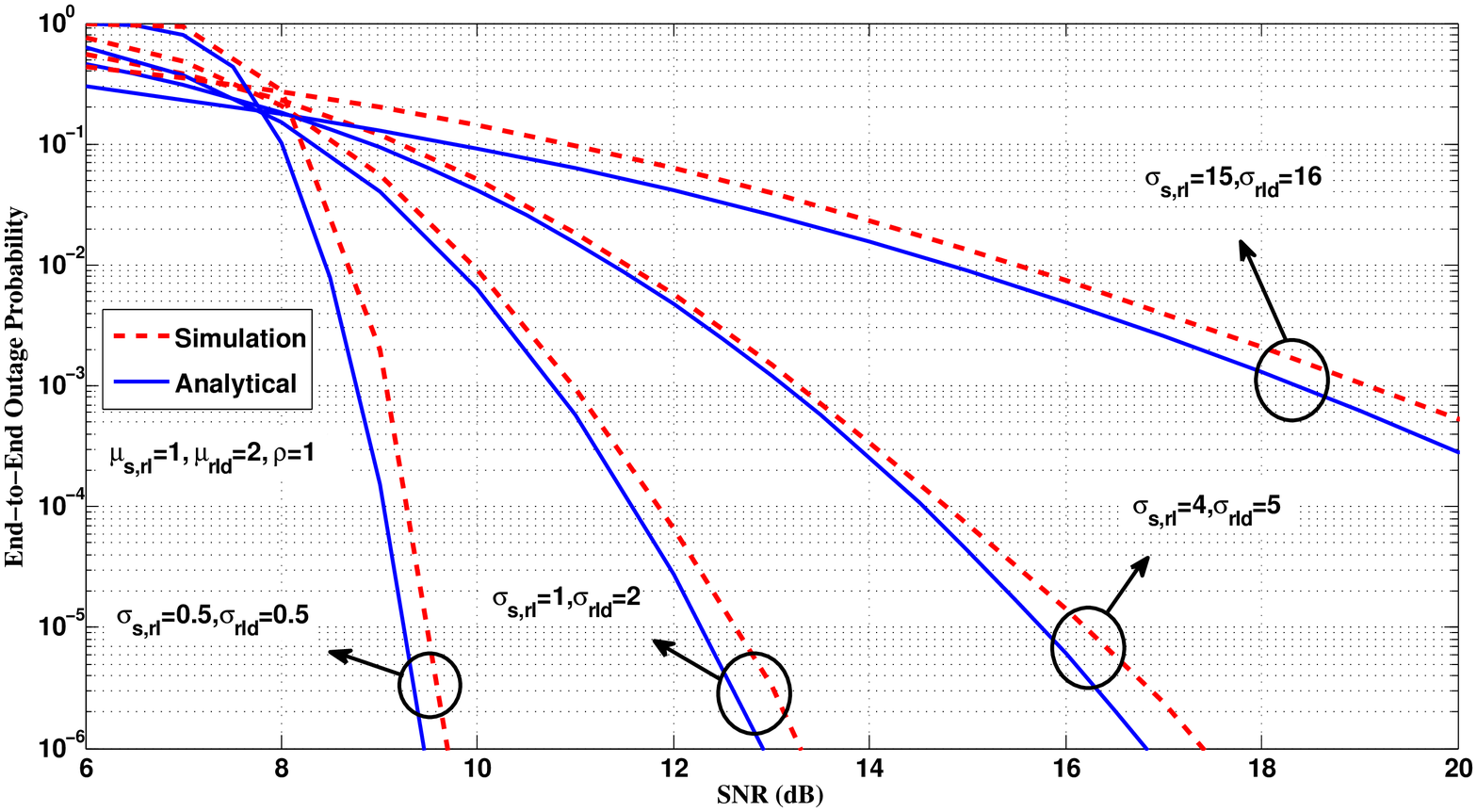}
\vspace{-1mm} \caption{ Comparing the accuracy of the analytical
outage probability derived for OR-AF using Wilkinson method with the
monte-carlo simulation.  } \label{fig: AF wilkinson} \vspace{-2mm}
\end{figure}

\begin{figure}[h]
\centering
\includegraphics[width=\onecolumnfigurewidth]{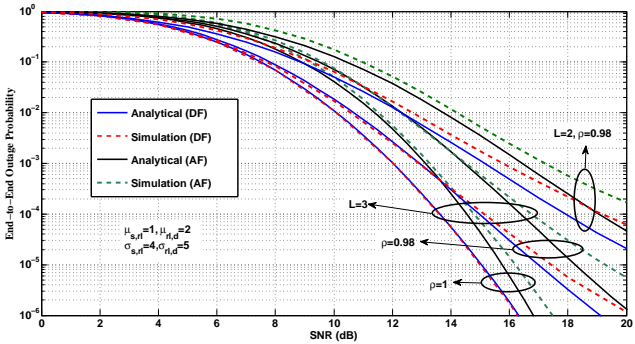}
\vspace{-1mm} \caption{ Comparison between analytical and simulated
outage probability of OR-DF and OR-AF scheme with imperfect CSI with
different number of relays. ($L=2$, $3$)  } \label{fig: AF_DF}
\vspace{-4mm}
\end{figure}
Fig.\ref{fig: OR-DF} and Fig.~\ref{fig: OR-AF} depict the effect of
channel estimation error on OR-DF and OR-AF protocol over Log-normal
fading channel, respectively. $\sigma_{l_1}$, $\mu_{l,1}$,
$\sigma_{l_2}$, $\mu_{l_2}$ are the parameters of the first and
second hop, respectively. Following conclusions are drawn from
Figs.~\ref{fig: OR-DF} and \ref{fig: OR-AF} :

\begin{enumerate}
\item As the $\rho$ decreases, or equivalently, the estimation error increases, the
performance become worse.
 \item As $\rho$ decreases the distance
between the simulation and analytical result increases.

\item Depending on $\rho$, in a specific SNR, the performance of the system
saturates, that is by increasing SNR, we do not get improvement in
the system performance.

\end{enumerate}

The second conclusion, is the outcome of the approximation that was
mentioned in section~\ref{sec:Model Estim}. However, simulations
show that approximating the estimated channel as a Log-normal R.V is
acceptable. We observe that for $\rho\simeq1$, the analytical result
is acceptable for good range of SNRs.

Fig.~\ref{fig: AF wilkinson} investigate the accuracy of Wilkinson
method with different variances for OR-AF system. As it is shown
in~\cite{Schwar:Bell:1982}, we see that the higher the variance, the
higher the difference between the parameters of the approximated pdf
according to Wilkinson method and the pdf according to simulation.

As it has been mentioned in~\cite{Hasna:WCOM:2003}, we can see that
in Fig.~\ref{fig: AF_DF} the DF outperforms the AF protocol in low
SNRs; however, they become close to each other in high SNRs. We can
also notice that the number of cooperating relays has a strong
impact on the performance enhancement, that is the achieved
diversity order is related to relay number, $L$.
\vspace{-0.5em}
\section{Conclusion}
\label{sec:con}
In this paper, end-to-end outage probability of a wireless
 communication system using best-worse and best-harmonic mean selection Over Imperfect Non-identical Log-normal Fading Channels has been investigated.
Channel estimation error has been considered as a Gaussian random
variable. Since the distribution of estimated channels are not
Log-normal either as would be in the case of the Rayleigh fading
channels, simulation results do not exactly follow the analytical
results. However, this difference is negligible for a wide variety
of situations. Harmonic mean of two Log-normal
 R.Vs in presence of channel estimation errors was also derived.

\section{Acknowledgements}
This work has been supported by Iran Telecommunication Research
Center (ITRC).


\end{document}